\documentclass[final,nopreprintline,12pt]{elsarticle}
\usepackage{xcolor}
\usepackage[utf8]{inputenc}
\usepackage{amsthm}
\usepackage{graphicx}

\usepackage{amsmath}  
\usepackage{amssymb} 

\usepackage{url}

\newcommand{\eps}{\varepsilon}

\newcommand{\N}{{\mathbb N}}

\newcommand{\IE}{{\mathbb E}}
\newcommand{\IP}{{\mathbb P}}

\newcommand{\II}{\mathcal{I}}

\newcommand{\hp}{{\hat p}}
\newcommand{\tp}{{\tilde p}}

\newcommand{\hxi}{{\hat \xi}}
\newcommand{\htau}{{\hat \tau}}

\journal{Journal of Parallel and Distributed Computing}

\newcommand{\1}[1]{{\mathbf 1}{\{#1\}}}

\newtheorem{theo}{Theorem}[section]
\newtheorem{lem}[theo]{Lemma}

\newtheorem{prop}[theo]{Proposition}
\newtheorem{cor}[theo]{Corollary}

\begin{document}

\begin{frontmatter}
\title{\textbf{FPC-BI: Fast Probabilistic Consensus 
within Byzantine Infrastructures}}

\author{Serguei Popov [1,2]
and William J Buchanan [2,3]}

\address{[1] Centro de Matem\'atica, University of Porto, Porto, Portugal}
\address{[2] IOTA Foundation, c/o Nextland, Strassburgerstraße 55, 10405 Berlin, Germany}
\address{[3] Blockpass ID Lab, Edinburgh Napier University, Edinburgh, UK}

\begin{abstract}
This paper presents a novel leaderless protocol (FPC-BI: Fast Probabilistic Consensus within Byzantine Infrastructures) with a low communicational complexity and which allows a set of nodes to come to a consensus on a value of a single bit. The paper makes the assumption that part of the nodes are \emph{Byzantine}, and are thus controlled by an adversary who intends to either delay the consensus, or break it (this defines that at least a couple of honest nodes come to different conclusions). We prove that, nevertheless, the protocol works with high probability when its parameters are suitably chosen. Along this the paper also provides  explicit estimates on the probability that the protocol finalizes in the consensus state in a given time. This protocol could be applied to reaching consensus in decentralized cryptocurrency systems. A special feature of it is that it makes use of a sequence of random numbers which are either provided by a trusted source or generated by the nodes themselves using some decentralized random number generating protocol. This increases the overall trustworthiness of the infrastructure. A core contribution of the paper is that it uses a  \emph{very weak} consensus to obtain a \emph{strong} consensus on the value of a bit, and which can relate to the validity of a transaction. 
\end{abstract}


\end{frontmatter}

\section{Introduction}
\noindent Increasingly, distributed systems need to provide a consensus on the current state of the infrastructure within given time limits, and to a high degree of accuracy. At the core of cryptocurrency transactions, for example, is that miners must achieve a consensus on the current state of transactions. This works well when all the nodes are behaving correctly, but a malicious agent could infect the infrastructure, and try and change the consensus \cite{zheng2017overview}. 

Suppose that there is a network composed of~$n$ nodes, and these nodes need to come to consensus on the value of a bit. Some of these nodes, however, may belong to an adversary, an entity which aims to delay  the consensus or prevent it from happening altogether. 
This paper focuses on this situation - and which is typical in the cryptocurrency applications - when the number~$n$ of nodes is large, and where they are possibly (geographically) spread out. This makes the communicational costs important whereas  computational complexity and the memory usage are often of a lesser concern. 

\subsection{Key contributions}
\noindent The key contribution of this paper is a protocol which allow a 
large number of adversarial nodes which may be a (fixed) proportion of the total number of nodes, while keeping the communicational complexity low (see Corollary~\ref{c_log_n}). It then guarantees fast convergence for all initial conditions. 
It is important to note that
here we do not require that with high probability the consensus should be achieved on the initial majority value. Rather, what we need, is:

\begin{itemize}
 \item[(i)] if, initially, no \emph{significant majority}\footnote{loosely speaking, a significant majority is something statistically different from the $50/50$ situation; for example, the proportion of $1$-opinion is greater than~$\alpha$ for some fixed $\alpha>1/2$} 
 of nodes prefer~$1$,
 then the final consensus should be~$0$ 
whp\footnote{``whp''
 $=$ ``with high probability''};
 \item[(ii)] if, initially, a 
 \emph{supermajority}\footnote{again, this is 
 a loosely defined notion; 
 a supermajority is something already close
 to consensus, e.g.\ more than 90\% of all nodes 
 have the same opinion}
 of nodes prefer~$1$,
 then the final consensus should be~$1$ whp.
 \label{items(i)-(ii)}
\end{itemize}

Along with these assumptions, another important assumption that we make is that, among the totality  of~$n$ nodes, there are $qn$ adversarial (Byzantine) nodes\footnote{where $q\in [0,1)$,
although
we will typically assume below that $q$ is less than certain threshold, which 
is at most $1/2$}, 
who may not follow the proposed protocol and who may
act maliciously in order to prevent the consensus (of the honest nodes) from being achieved.

\subsection{Context}
\noindent To understand the importance of this work to cryptocurrency applications, consider a situation when there are two contradicting transactions. For example, if one transfers all the balance of address~$A_1$ to address~$A_2$, while the other transfers all the balance of address~$A_1$ to address~$A_3\neq A_2$. If those two transactions appear roughly at the same time, then neither of the two transactions will be strongly preferred by the nodes of the network, 
 they can then be declared invalid - just in case. 
On the other hand, it would not be a good idea to \emph{always} declare them invalid, as a malicious actor (Eve) could be able to exploit this. For example, Eve could place a legitimate transaction, such as buying some goods from a merchant. When she receives the goods, she publishes a double-spending transaction - as above - in the hope that \emph{both} will be canceled, and so she would effectively receive her money back (or at least take the money away from the merchant). To avoid this kind of threat, it would be desirable if the first transaction (payment to the merchant) which, by that time, would have probably gained some confidence from the nodes, would stay confirmed, and only the subsequent double-spend gets canceled.

\section{Related Work}
\noindent There is a wide range of classical work on (probabilistic) Byzantine consensus protocols \cite{aguilera2012correctness,ben1983another,bracha1987asynchronous,feldman1989optimal,friedman2005simple,rabin1983randomized}. The disadvantage of the approach of these papers is, however, that they typically require that the nodes exchange $O(n^2)$ messages in each round (which means $O(n)$ messages for each node). In the situation where the communicational complexity matters, this can be a major barrier. 

A good deal of work focuses on failures within a network infrastructure, rather than on malicious agents. The work of Liu \cite{liu2018scalable} defines  FastBFT, and which is a fast and scalable BFT (Byzantine fault tolerance) protocol. Within this, the work integrates trusted execution environments (TEEs) with lightweight secret sharing, and results in a low latency infrastructure. Crain et al \cite{crain2018dbft} define Democratic Byzantine Fault Tolerance (DBFT) and which is a leaderless Byzantine consensus. This provides a robust infrastructure where there is a failure in the leader of the consensus network. The core contribution is that nodes will process message whenever they receive them, instead of waiting for a co-ordinate to confirm messages. Another Byzantine Fault Tolerant method which does not require a leader node is Honey Badger \cite{miller2016honey}. This method is asynchronous in its scope and can cope with corrupted nodes. Unfortunately, it does not actually make any commitments around the timing of the delivery of a message, and where even if Eve controls the scheduling of messages, there will be no impact on the overall consensus. 

There has also been much research on the probabilistic models where, in each round, a node only contacts a small number of other nodes in order to learn their opinions, and possibly change its own. This type of models is usually called \emph{voter models}, and which were introduced in the 70s by Holley and Liggett~\cite{holley1975ergodic} and Clifford and Sudbury~\cite{CliffSud}.
A very important observation is that, in most cases, voter models have only two extremal 
invariant measures: one concentrated 
on the ``all-$0$'' configuration, and the other one concentrated
on the ``all-$1$'' -- we can naturally call these two configurations ``consensus states''. 
Since then, there has been a range of work on voter models;
in particular, let us cite 
\cite{becchetti2016stabilizing, cooper2014power,cooper2015fast,elsasser2016rapid,fanti2019communication,CruiseGanesh14} which are specifically
aimed at reaching consensus and have low communicational
complexity (typically, $O(n\ln n)$). 
However, in these works, the presence of adversarial nodes is usually either not allowed,
or is supposed to be very minimal 
(the system can admit only roughly $O(\sqrt{n})$ adversarial nodes, so the 
allowed proportion of the adversarial nodes is asymptotically zero).

\section{Model Definition}
\label{s_descr}


The developed model assumes that adversarial nodes can exchange information freely between themselves and can agree on a common strategy. In fact, they all may be controlled by a single individual or entity. We also assume that the adversary is \emph{omniscient}: at each moment of time, he is aware of the current opinion of every honest node. While this assumption may seem a bit too extreme, note that the adversarial  nodes can query the honest ones a bit more frequently to be aware of the current state of the network;
also, even if the ``too frequent'' queries are somehow not permitted, the adversary can still \emph{infer} (with some degree of confidence) about the opinion of a given honest node by analyzing the history of this node's interactions with all the adversarial nodes.

The remaining $(1-q)n$ nodes are \emph{honest}, i.e.,
they follow the recommended protocol. We assume that 
they are numbered from~$1$ to~$(1-q)n$;
this will enter into several notations below.

Our protocol will be divided into epochs which we
call rounds; for now, let us assume that the end of the previous round (which
coincides with the beginning of the next round) occur at predetermined 
time instances.\footnote{in fact, those instances are the times when the common
random numbers are made available to the system, see Section~\ref{s_generating_RNs} below}
The basic feature of it
is that, in each round, each node may query~$k$ other nodes
about their current opinion (i.e., the preferred value of the bit).  We allow~$k$ to be relatively large (say, $k=50$ or so), but still assume that $k\ll n$.
We also assume that the complete list of the nodes is known to all the participants, and any node can  \emph{directly} query any other node. For the sake of clarity of the presentation,
for now we assume that all nodes (honest and adversarial) always respond to the queries; in Section~\ref{s_gener} we deal with the general situation when nodes can 
possibly remain silent. This, by the way, will result in
a new ``security threshold''~$\phi^{-2}\approx 0.38$ 
(where~$\phi$
is the Golden Ratio), different from the ``usual''
security thresholds~$\frac{1}{2}$ and~$\frac{1}{3}$.

With respect to the \emph{behavior} of the adversarial nodes, there are two important cases to be distinguished:
\begin{itemize}
 \item \emph{Cautious adversary}\footnote{also know as
\emph{covert adversary}, cf.~\cite{aumann2007security}}: 
any adversarial node
must maintain the same opinion in the same round,
i.e., respond the same value to all the queries
it receives in that round.
 \item \emph{Berserk adversary}:
 an adversarial node may respond differently to things for different queries in the same round.
\end{itemize}

To explain the reason why the adversary may choose to be cautious,
first note that we also assume that
nodes have identities and sign all their messages;
this way, one can always \emph{prove} that a given message originates from a given node.
Now, if a node is not cautious, 
this may be detected by the honest nodes (e.g., two  honest nodes may exchange their query history and verify that the same node passed contradicting information to them).
In such a case, the offender may be penalized
by \emph{all} the honest nodes (the nodes who discovered
the fraud would pass that information along, together
with the relevant proof%
). Since, in the sequel, we will see that the protocol provides more security and converges faster  against a cautious adversary, it may be indeed a good idea for the honest nodes to adopt additional measures in order to detect the ``berserk'' behavior.
Also, since~$k$ would be typically large and each node
is queried~$k$ times on average during each round,
we make a further simplifying assumption that
a cautious adversary just chooses (in some way) the opinions
of all his nodes \emph{before} the current round starts and then communicates these opinions to whoever asks.

\subsection{Random numbers}
\label{s_generating_RNs}
\noindent The protocol we are going to describe requires
the system to 
use, from time to time
(more precisely, once in each round),
a random number available to all the participants
(this is very similar to the  ``global-coin'' approach
used in many works on Byzantine
consensus, see e.g.~\cite{aguilera2012correctness}).
For the sake of cleanness of the presentation
and the arguments, in this paper we mainly assume that these
random numbers
are provided by a trusted source, not controlled
by the adversary\footnote{i.e., the adversary 
may be omniscient (knows all information
that exists \emph{now}), but he is not \emph{prescient}
(cannot know the future)}. 
We observe that such random number generation
can be done in a decentralized way 
as well (provided that the proportion~$q$
of the adversarial nodes is not too large),
see e.g.~\cite{cascudo2017scrape, Lenstra_Wes17, popov2017decentralized, schindlerhydrand, syta2017scalable}.
If a ``completely decentralized'' solution
proves to be too expensive (from the point of view
of computational and/or communicational complexity),
one can consider ``intermediate'' ones, such as 
using a smaller committee for this, and/or making
use of many publicly available RNGs.
It is important to observe that (as we will see from the 
analysis below), 
even if from time to time the adversary can get
(total or partial) control 
of the random number, this can only lead to delayed consensus,
but he cannot convince different honest nodes of different 
things, i.e., safety is not violated.
Also, it is not necessary that really \emph{all}
honest nodes agree on the same number; if most 
of them do, this is already fine.
This justifies the idea that, in our context, both decentralization
and ``strong consensus''
are not of utter importance for the specific task
of random number generation.
We postpone the rest of this discussion
to Section~\ref{s_gener}, since the methods
 we employ for proving our results are relevant for it.

Before actually describing our protocol,
it is important to note that we assume that there is 
no central entity that ``supervises'' the network
and can somehow know that the consensus was achieved
and therefore it is time to stop. This means that 
each node must decide when to stop using a \emph{local}
rule, i.e., using only the information locally available
to it. A quick example of such a decision rule might be the following:
``if during the last~$10$ rounds at least $2/3$ of my queries were answered 
with opinion~`$1$', then I consider this opinion to be final''.

\subsection{Parameter setup}
\label{s_parsetup}
\noindent The protocol depends on a set of integer and real parameters:
\begin{itemize}
 \item $1/2 < a \leq b <1$, the threshold limits
 in the first round;\footnote{they are needed to assure (i)--(ii)
on page~\pageref{items(i)-(ii)}; $a$ has to be larger than the ``significant
majority threshold'', while $b$ needs to be smaller than the ``supermajority
threshold''}
 \item $\beta\in (0,1/2)$, the threshold limit
 parameter in the subsequent rounds;
 \item $m_0\in\N$, the cooling-off period;
 \item $\ell\in\N$, the number of consecutive rounds
 (when the cooling-off period is over)
 with the same opinion after which it becomes final, 
 for one node. 
\end{itemize}
Now, let us describe our protocol.
First, we assume that each node 
decides on the initial value of the bit,
according to any reasonable rule\footnote{for example,
if a node sees a valid transaction~$x$ (which 
does not contradict to prior transactions) at time~$t$,
and during the time interval $[t,t+\Delta]$
it does not see any transactions that contradict to~$x$
it may initially decide that~$x$ is \emph{good},
setting the value of the corresponding bit to~$1$}.
Then, we describe the \emph{first round} of the  protocol in the following way:
\begin{itemize}
 \item in the first round,
 each honest node~$j$ 
 randomly queries other nodes~$k$ times
(repetitions and self-queries are allowed\footnote{we 
have chosen this mainly to facilitate the 
subsequent analysis; of course, e.g.\ querying~$k$
\emph{other} nodes chosen uniformly at random can be analyzed in a similar way, with some more technical complications because the relevant random variables will not be exactly Binomial, but only approximately so
(it will be Hypergeometric, in fact). In any case,
we will see below that~$k$ typically will be much less 
than~$\sqrt{n}$, which makes querying the same node 
in the same round quite unlikely.})
and records the number~$\eta_1(j)$ of $1$-opinions
it receives;
\item \emph{after that}, 
 the value of the random variable
$X_1\sim U[a,b]$ is made available to the nodes\footnote{$U[a,b]$
stands for the uniform probability distribution on interval $[a,b]$};
\item then, each honest node
 uses the following decision rule:
if $k^{-1}\eta_1(j)\geq X_1$, it adopts opinion~$1$,
otherwise 
it adopts opinion~$0$.
\end{itemize}

In the subsequent rounds, the dynamics is almost the 
same, we only change the interval where the uniform 
random variable lives:
\begin{itemize}
 \item in the round~$m\geq 2$,
 each honest node~$j$ 
randomly queries other nodes~$k$ times,
and records the number~$\eta_m(j)$ of $1$-opinions
it receives;
\item \emph{after that}, 
 the value of the random variable
$X_m\sim U[\beta,1-\beta]$ is made available to the nodes;
\item then, each honest node which does not
yet have \emph{final} opinion
 uses the following decision rule:
if $k^{-1}\eta_m(j)\geq X_m$, it adopts opinion~$1$,
otherwise it adopts opinion~$0$.
\end{itemize}
As mentioned above, if an honest node has the same 
opinion during~$\ell$ consecutive rounds
after the cooling-off period
(i.e., counting from time $m_0+1$ on)
this opinion becomes final. 

\subsection{Consensus mechanism}
\noindent Let us now explain informally what makes our protocol converge fast to the consensus even in the Byzantine setting. The general idea is the following: if the adversary (Eve) knows the decision rules that the honest nodes use, she can then predict their behaviour and adjust her strategy accordingly, in order to be able to 
delay the consensus and further mess with the system. Therefore, let us make these rules  unknown to all the participants, including Eve. Specifically, even though Eve's nodes can control (to some extent) the expected proportion of $1$-responses among the~$k$ queries,
she cannot control the value that the ``threshold'' random
variable assumes. As a consequence, the decision
threshold~$X_1$ will likely be ``separated'' from
that typical proportion.

When this separation happens, the opinions of the honest nodes would tend very strongly in one of the directions whp. Then, it will be extremely unlikely that 
the system leaves this ``pre-consensus'' state,
due to the fact that the decision thresholds, however
random, are always uniformly away from~$0$ and~$1$.
Also, we mention that a similar protocol with intended cryptocurrency applications was 
considered in~\cite{rocket2019ava}. However, there 
only ``fixed thresholds'' were used, 
which gives Eve much more control, so that, in particular, then she could delay the consensus a great deal.
As a last remark, it is important to note that having
``independently random thresholds'' (i.e., each node 
independently chooses its own decision threshold) is not
enough to achieve the effect described above
--- these ``locally random'' decisions will simply average out;
 that is, having common random numbers is indeed essential.


\section{Results}
\label{s_results}
\noindent We define two events relative to the final consensus
value:
\begin{equation}
 H_i = \{ \text{all honest nodes eventually reach
   final opinion } i\},
 \quad i=0,1.
\end{equation}
Thus, the union $H_0\cup H_1$ stands for the event
that all honest nodes agree on the same value,
i.e., that the consensus was achieved.

For $0<q< \beta<\frac{1}{2}$, abbreviate 
\[
\varphi_{\beta,q,k}=\frac{\beta-q}{2(1-q)}
-e^{-\frac{1}{2}k(\beta-q)^2}.
\]
In the following, 
we assume that~$k$ is large enough so that 
$\varphi_{\beta,q,k}>0$
(indeed, the first term in the above display
is strictly positive, and the second
one converges to~$0$ as $k\to\infty$).
Let us also denote, for $n,k,m_0,\ell$ as above
and a generic non negative 
integer~$u$
\begin{align}
 W(n,k,m_0,\ell,u) & = (1-q)n\Big(\big(1
  -\big(1-e^{-\frac{1}{2}k(\beta-q)^2}\big)^\ell\big)^u
    + \Big(\frac{e^{-\frac{1}{2}k(\beta-q)^2}}
    {1-e^{-\frac{1}{2}k(\beta-q)^2}}\Big)^{\ell-1}\Big)
 \nonumber\\
&\qquad 
+(m_0+\ell u)
 e^{-2(1-q)n\varphi_{\beta,q,k}^2},
\label{df_bigW}
\end{align}
and
\begin{align}
 \psi_{cau}(n,k) &= 2\exp\big(-\tfrac{1}{8}n
     \tfrac{(\beta-q)^2}{4(1-q)}\big)
      + (1-2\beta)^{-1} \sqrt{2k^{-1}\ln\tfrac{4(1-q)}{\beta-q}},
\label{df_psi_cautious}
 \\
 \psi_{ber}(n,k) &= 2\exp\big(-\tfrac{1}{8}n
     \tfrac{(\beta-q)^2}{4(1-q)}\big)
    + \frac{q+\sqrt{2k^{-1}
\ln\tfrac{4(1-q)}{\beta-q}}}{1-2\beta}.
\label{df_psi_berserk}
\end{align}
(in the above notation, we omit the dependence on~$q$ and $\beta$).
As it will become clear shortly, we will need $W(n,k,m_0,\ell,u)$
to be small, and $\psi$'s 
(which, as the reader probably have noted, relate to cautious
and berserk adversaries) to be strictly less than~$1$.
It is not difficult to see (we elaborate more on that below)
that (recall that $q< \beta$) the value
of the expression in~\eqref{df_bigW} will
be small indeed if~$n$ is large and~$k$ is at least $C \ln n$
for a large~$C$ (indeed, with other parameters fixed,
note that each of the two terms in the second parentheses will be polynomially
small in~$k$, with the factor~$C$ entering to the negative power; 
a fixed $C>(\beta-q)^{-2}$ works for all large enough~$n$).
Then, the first term in the expression in~\eqref{df_psi_cautious}
will be very small for large~$n$, while the second term 
will also be small for a sufficiently large~$k$. 
As for~\eqref{df_psi_berserk}, it shares the same
first term with~\eqref{df_psi_cautious}; 
the second term, however, will be of constant order,
and if we want it to be strictly less than~$1$ for a 
large~$k$, we need the constraint $q<1-2\beta$
to hold.

Now, we begin formulating our main results.
Let $\mathcal{N}$ be the number of rounds until 
\emph{all} honest nodes achieve their final opinions.
The next result controls both the number of necessary
rounds and the probability that the final consensus
is achieved (i.e., the event $H_0\cup H_1$ occurs):
\begin{theo}
\label{t_safety_liveness}
\begin{itemize}
\item[(i)] 
For any strategy of a cautious adversary,
it holds that
\begin{equation}
\label{eq_safety_liveness_cautious} 
\IP\big[(H_0\cup H_1)\cap\{\mathcal{N}\leq m_0+\ell u\}\big]
 \geq 1-W(n,k,m_0,\ell,u)- \big(\psi_{cau}(n,k)\big)^{m_0} .
\end{equation}
\item[(ii)]
For any strategy of a berserk adversary, we have 
\begin{equation}
\label{eq_safety_liveness_berserk} 
\IP\big[(H_0\cup H_1)\cap\{\mathcal{N}\leq m_0+\ell u\}\big]
 \geq 1-W(n,k,m_0,\ell,u)- \big(\psi_{ber}(n,k) \big)^{m_0}.
\end{equation}
\end{itemize}
\end{theo}
Note that the first term in~\eqref{df_bigW} decreases in~$u$ while the second
one increases in~$u$; however, the second term is typically of much smaller 
value (since it has a exponentially small in~$n$ factor) and therefore
one may obtain better estimates in~\eqref{eq_safety_liveness_cautious}
and~\eqref{eq_safety_liveness_berserk} with some strictly positive values of~$u$.
Note also that the only difference 
between~\eqref{eq_safety_liveness_cautious}
and~\eqref{eq_safety_liveness_berserk}
is in the second terms of~\eqref{df_psi_cautious}
and~\eqref{df_psi_berserk}.
As we will see in the proofs, these terms enter
into the part which is ``responsible'' for the estimates on
the time moment when the adversary loses control on the situation
which permits one of the opinions to reach a supermajority;
from that moment on, there is essentially no difference
if the adversary is cautious or berserk.

\begin{cor}
\label{c_1/2_1/3}
 For a cautious adversary we need that $q<\beta$,
 while for a berserk adversary we \emph{also}
need that $q<1-2\beta$.
 Recalling also that~$\beta$ must belong to~$(0,1/2)$,
it is not difficult to see that
\begin{itemize}
 \item for a cautious adversary, for any~$q<1/2$
and all large enough~$n$ 
 we are able to adjust the parameters $k,\beta,m_0,\ell$
 in such a way
that the protocol works whp (in particular, 
a $\beta$-value sufficiently close to~$1/2$ would work);
 \item however, for a berserk adversary, we are able
 to do the same only for $q<1/3$ (here, $\beta=1/3$
 would work).
\end{itemize}
\end{cor}

\begin{cor}
\label{c_log_n}
 One may be interested in asymptotic results,
for example, of the following kind: assume that
the number of nodes~$n$ is fixed (and large), 
and the proportion of Byzantine nodes~$q$
is \emph{acceptable} (i.e., less than~$1/2$ for 
the case of cautious adversary, or less than~$1/3$
for the case of berserk adversary, as discussed above).
We then want to choose the parameters
of the protocol in such a way that the probabilities
in~\eqref{eq_safety_liveness_cautious} 
and~\eqref{eq_safety_liveness_berserk}
are at least $1-\eps(n)$, where~$\eps(n)$
is polynomially small in~$n$ (i.e., $\eps(n) = O(n^{-h})$ for 
some~$h>0$). 

First, $\beta=1/3$ works in both cases; then, 
a quick analysis 
of~\eqref{eq_safety_liveness_cautious}--\eqref{eq_safety_liveness_berserk}
shows that one possibility is: chose $k= C \ln n$
(with a sufficiently large constant in front), 
$\ell$ (the number of consecutive rounds with the same opinion
before finalization) of constant order, and
 $m_0 = O\big(\frac{\ln n}{\ln \ln n}\big)$
for cautious adversary or $m_0=O(\ln n)$
for a berserk one.

That is, the overall communicational complexity
will be at most $O\big(\frac{n \ln^2 n}{\ln \ln n}\big)$
for a cautious adversary
and $O(n \ln^2 n)$ for a berserk one.
\end{cor}

Next, let~$\hp_0$ be the initial proportion of $1$-opinions
among the honest nodes.
Our second result shows that if, initially, no  
significant majority of nodes prefer~$1$,
then the final consensus will be~$0$ whp, and
if the supermajority of nodes prefer~$1$,
 then the final consensus will be~$1$ whp
(recall (i)--(ii) 
on page~\pageref{items(i)-(ii)}),
and it is valid in the general case (i.e.,
for both cautious and berserk adversaries).
\begin{theo}
\label{t_initial}
\begin{itemize}
 \item[(i)]
First, suppose that $\hp_0 (1-q) + q < a$,
and assume that~$k$ is sufficiently large so that 
\[
 e^{-2k(a-\hp_0 (1-q) - q)^2} \leq \frac{\beta-q}{4(1-q)}. 
\]
Then, we have 
\begin{align}
\IP\big[H_0\cap\{\mathcal{N}\leq m_0+\ell u\}\big]
 &  \geq 1 - \exp\big(-\tfrac{1}{8}n
     \tfrac{(\beta-q)^2}{4(1-q)}\big)
      -W(n,k,m_0,\ell,u).
\label{eq_initial0}     
\end{align}
 \item[(ii)]
 Now, suppose that $\hp_0(1-q)>b$, 
and assume that~$k$ is sufficiently large so that 
\[
 e^{-2k(\hp_0 (1-q) - b)^2} \leq \frac{\beta-q}{4(1-q)}. 
\]
 Then, 
the same estimate~\eqref{eq_initial0} holds for
$\IP[H_1\cap\{\mathcal{N}\leq m_0+\ell u\}]$.
\end{itemize}
\end{theo}

We also mention that the estimates~\eqref{eq_safety_liveness_cautious} 
and~\eqref{eq_safety_liveness_berserk} are usually not quite sharp
because we have used some union bounds and other
``worst-case'' arguments when proving them. In practice, one 
might resort to simulations to possibly achieve better estimates;
since the present paper is mostly devoted to a theoretical analysis
of the protocol in the simplest setting, 
we defer this discussion to Section~\ref{s_final}.
For a quick concrete example on the number of necessary
rounds until consensus (with different parameters), 
see Figure~\ref{f_simul_fpc_seb} (courtesy of Sebastian M\"uller).
\begin{figure*}
\begin{center}
\includegraphics[width=\textwidth]{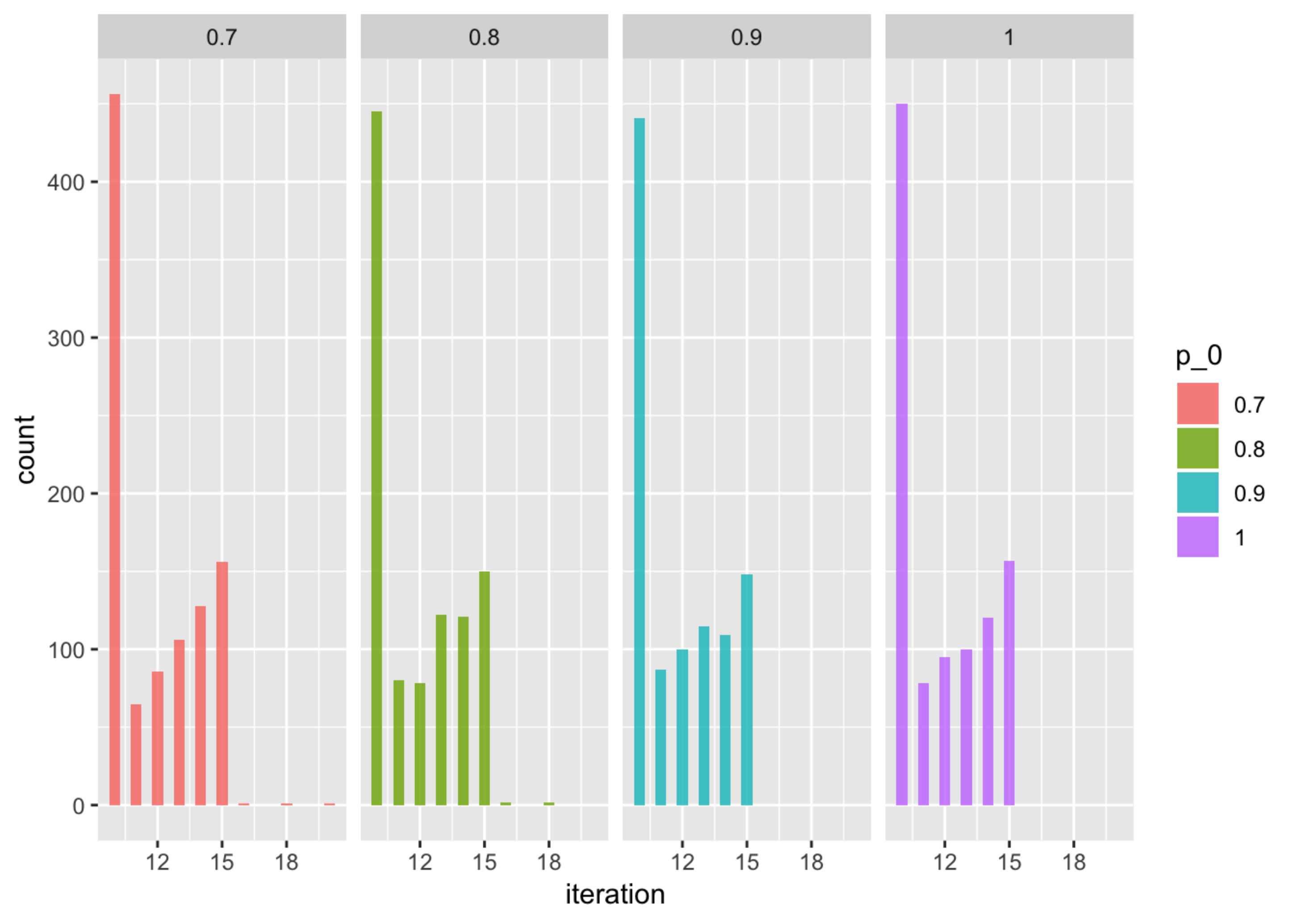}
\caption{Number of rounds till the protocol
finalizes, with $n=1000$, $a=0.75$, $b=0.85$,
$m_0=\ell=5$, $k=20$, and $q=0.1$. Here, $p_0$ is the initial proportion
of $1$-opinions among honest nodes; the round counts are on the horizontal axes,
and on the vertical axes are numbers of nodes that have finalized
on the $1$-opinion after that number of rounds.}
\label{f_simul_fpc_seb}
\end{center}
\end{figure*}
It is interesting to observe that, in most cases,
the protocol finalizes after the minimal
number $m_0+\ell=10$ of rounds
and the probability that it lasts for more 
than~$20$ rounds seems to be very small.
%
%
%
%
%

Before starting with the proofs, let us mention that, for reader's convenience, we provide the list of most important notations
of this paper in Section~\ref{s_notations} (page~\pageref{s_notations}).

\section{Proofs}
\label{s_proofs}
\noindent We start with some preliminaries.
Denote by $\mathcal{B}(m,p)$ the Binomial distribution 
with parameters $m\geq 1$ and $p\in [0,1]$.
Let us recall
the Hoeffding's inequality~\cite{hoeffding1994probability}:
if $\gamma\in (0,1)$ is a parameter and $S_m\sim \mathcal{B}(m,p)$ 
with $0<\gamma<p<1$, then
\begin{equation}
\label{H_Hoef}
\IP\big[m^{-1} S_m \leq \gamma\big] \leq \exp\{-2m(p-\gamma)^2\};
\end{equation}
the same estimate also holds
for $\IP[m^{-1} S_m \geq \gamma]$ in the case $0<p<\gamma<1$.

\begin{figure}
\begin{center}
\includegraphics{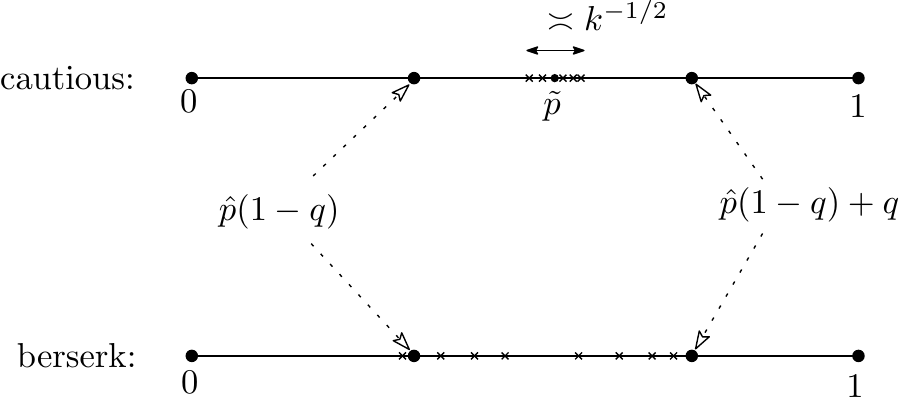}
\caption{What cautious and berserk adversaries can achieve}
\label{f_cautious_berserk}
\end{center}
\end{figure}
To better understand the difference between
cautious and berserk adversaries,
look at Figure~\ref{f_cautious_berserk}.
Here, $\hp$ is the initial proportion
of $1$-opinions between the honest nodes, 
and the crosses mark the proportion 
of $1$-responses to the~$k$ queries that the honest
nodes obtain. The cautious adversary 
can choose any $\tp \in [\hp(1-q),\hp(1-q)+q]$
(by adjusting the opinions of his nodes appropriately,
so that the overall proportion of $1$-opinions
would be~$\tp$),
and then those crosses will be (mostly) concentrated in the 
interval of length of order $k^{-1/2}$ around~$\tp$.
On the other hand, the berserk adversary
can cause the crosses to be distributed in any way
on the whole interval $[\hp(1-q),\hp(1-q)+q]$,
with some of them even going a bit out of it
(on the distance of order $k^{-1/2}$ again).

Next, we need an auxiliary result on
a likely outcome of a round in the case
when the adversary cannot make the typical
proportion of $1$-responses to be close to 
the decision threshold. Let~$\eta(j)$
be the number of $1$-responses among~$k$
queries that $j$th honest node receives;
in general, the random variables
$(\eta(j),j=1,\ldots,(1-q)n)$
are not independent, but they are 
\emph{conditionally independent} given the 
adversary's strategy. (Note that
$\eta(j)\sim \mathcal{B}(k,\tp)$
with some possibly random~$\tp$ if the adversary is cautious,
but the situation may be more complicated 
for a berserk one.) For a fixed~$\lambda\in (0,1)$, define a random variable
\[
 \hp = \frac{1}{(1-q)n} 
   \sum_{j=1}^{(1-q)n}\1{\eta(j)\geq \lambda k};
\]
so that~$\hp$ is the new proportion of $1$-opinions
among the honest nodes, given that
the ``decision threshold'' equals~$\lambda$. 
Then, the following result holds:
\begin{lem}
\label{l_LD_fixed}
\begin{itemize}
 \item[(i)] Assume that,
 conditioned on any adversarial strategy,
 there are some positive~$c$ and~$\theta$ such that 
 $\eta(j)$ is stochastically
 dominated by $\mathcal{B}(k,\lambda-c)$
 for all $j=1,\ldots,(1-q)n$,
and $\IP[\mathcal{B}(k,\lambda-c)\geq \lambda k]\leq\theta$. Then,
for any $v>0$
\begin{equation}
\label{eq_LD0}
\IP[\hp>\theta+v] \leq e^{-2(1-q)nv^2}.
\end{equation}
 \item[(ii)]  Assume that,
 conditioned on any adversarial strategy,
 $\eta(j)$ stochastically
 dominates $\mathcal{B}(k,\lambda+c)$
 for all $j=1,\ldots,(1-q)n$,
and $\IP[\mathcal{B}(k,\lambda+c)\leq \lambda k]\leq \theta$. Then,
for any $v>0$
\begin{equation}
\label{eq_LD1}
\IP[\hp<1-\theta-v] \leq e^{-2(1-q)nv^2}.
\end{equation}
\end{itemize}
\end{lem}

\begin{proof}
For~(i), we observe that $(1-q)n\hp$ is 
stochastically dominated by $\mathcal{B}(n,\theta)$,
and then~\eqref{eq_LD0} follows from~\eqref{H_Hoef}.
The proof of the part~(ii) is completely analogous.
\end{proof}
Note that, by~\eqref{H_Hoef},
$\IP[\mathcal{B}(k,\lambda-c)\geq \lambda k]\leq e^{-2kc^2}$
(and the same holds for 
$\IP[\mathcal{B}(k,\lambda+c)\leq \lambda k]$),
so we will normally use Lemma~\ref{l_LD_fixed}
with $\theta = e^{-2kc^2}$.

Another elementary fact we need is
\begin{lem}
\label{l_elem_calc}
 Let $(\xi_m^{(j)}, m\geq 1), j=1,\ldots,N$ be~$N$
sequences of independent Bernoulli 
trials\footnote{the sequences themselves are \emph{not}
assumed to be independent between each other}
with success probability $h\in (0,1)$.
For $j=1,\ldots,N$ define 
\[
 \tau^{(1)}_j = \min\big\{m\geq \ell : 
  \xi_m^{(j)}=\xi_{m-1}^{(j)}=\ldots 
  = \xi_{m-\ell+1}^{(j)}=1\big\}
\]
and
\[
 \tau^{(0)}_j = \min\big\{m\geq \ell : 
  \xi_m^{(j)}=\xi_{m-1}^{(j)}=\ldots 
  = \xi_{m-\ell+1}^{(j)}=0\big\}
\]
to be the first moments when runs of~$\ell$
ones (respectively, zeros) are observed in $j$th
sequence. 
Then, for all $u\in \N$,
\begin{equation}
\label{eq_elem_calc}
 \IP[\tau^{(1)}_j\leq \ell u, \tau^{(1)}_j<\tau^{(0)}_j]
  \geq 1 - (1-h^\ell)^u - \Big(\frac{1-h}{h}\Big)^{\ell-1}
\end{equation}
for all $j=1,\ldots,N$,
and
\begin{equation}
\label{eq_elem_calc_all}
 \IP[\tau^{(1)}_j\leq \ell u, \tau^{(1)}_j<\tau^{(0)}_j,
 \forall j=1,\ldots,N]
  \geq 1 - N\Big((1-h^\ell)^u +
    \Big(\frac{1-h}{h}\Big)^{\ell-1}\Big).
\end{equation}
\end{lem}
\begin{proof}
First, it is clear that
\begin{equation}
\label{tau_blocks}
 \IP[\tau^{(1)}_j \leq \ell u] \geq 1 - (1-h^\ell)^u
\end{equation}
(divide the time interval $[1,\ell u]$ into~$u$
subintervals of length~$\ell$ and note that
each of these subintervals is all-$1$
with probability~$h^\ell$).
Then, the following is an easy exercise on
computing probabilities via conditioning
(for the sake of completeness, we prove this
fact in the Appendix):
\begin{equation}
\label{runs}
 \IP[\tau^{(1)}_j<\tau^{(0)}_j] = 1
  - \frac{(1-h)^{\ell-1}(1-h^\ell)}
  {h^{\ell-1}+(1-h)^{\ell-1}-(h(1-h))^{\ell-1}}.
\end{equation}

Observe that~\eqref{runs} implies that
(since $1-h^\ell\leq 1$ and 
$(1-h)^{\ell-1}-(h(1-h))^{\ell-1}\geq 0$)
\[
\IP[\tau^{(1)}_j<\tau^{(0)}_j]
\geq 1 - \Big(\frac{1-h}{h}\Big)^{\ell-1},
\]
and so, using the above together with~\eqref{tau_blocks}
and the union bound, we obtain~\eqref{eq_elem_calc}.
The relation~\eqref{eq_elem_calc_all}
is then a direct consequence of~\eqref{eq_elem_calc}
(again, with the union bound).
\end{proof}


%

To prove our main results, we need some additional notation.
Let~$\varrho(j)$ be the round when the $j$th (honest)
node finalizes its opinion. 
Denote
\[
 R_m = \{ j: \varrho(j)\leq m\}
\]
to be the subset of honest nodes that finalized
their opinions by round~$m$. 
Let also~$\hxi_m(j)$ be the opinion of $j$th node after the 
$m$th round and  
\begin{equation}
\label{def_hp_m}
 \hp_m = \frac{1}{(1-q)n}
  \sum_{j=1}^{(1-q)n} \hxi_m(j)
\end{equation}
 be the proportion of $1$-opinions among the honest nodes after the $m$th round
in the original system.
\begin{proof}[Proof of Theorem~\ref{t_safety_liveness}.]
Let us define the random variable
\begin{equation}
\label{def_Psi}
 \Psi = \min\Big\{m\geq 1: \hp_m\leq \frac{\beta-q}{2(1-q)}
 \text{ or }\hp_m\geq 1-\frac{\beta-q}{2(1-q)}\Big\}
\end{equation}
to be the round after which 
the proportion of $1$-opinions among the honest
nodes either becomes ``too small'', or ``too large''.
We now need the following fact:
\begin{lem}
\label{l_dominate_Psi}
For all~$s\leq m_0+\ell$,
 it holds that (recall~\eqref{df_psi_cautious} 
 and~\eqref{df_psi_berserk})
\begin{equation}
\label{eq_dominate_Psi}
 \IP[\Psi>s] \leq 
  \begin{cases}
    \big(\psi_{cau}(n,k)\big)^{s-1}, &
       \text{for cautious adversary},\\
        \big(\psi_{ber}(n,k)\big)^{s-1},
        \vphantom{\int\limits^{B}}&
       \text{for berserk adversary}.
  \end{cases}
\end{equation}
\end{lem}
\begin{proof}
Observe that $s\leq m_0+\ell$ implies that a node
cannot finalize its opinion before round~$s$.
Consider first the case of a cautious adversary.
Abbreviate (for this proof) $\mu=\frac{\beta-q}{4(1-q)}$.
Let $m\geq 2$ and observe that,
for \emph{any} fixed $h\in [0,1]$ we have
(recall that $X_m\sim U[\beta,1-\beta]$)
\begin{align}
 \IP\big[e^{-2k(X_m-h)^2}\geq \mu\big]
 &= \IP\Big[(X_m-h)^2\leq\frac{\ln\mu^{-1}}{2k}\Big]
 \nonumber\\
 &= \IP\Big[h-\sqrt{\frac{\ln\mu^{-1}}{2k}} \leq
X_m\leq h+\sqrt{\frac{\ln\mu^{-1}}{2k}}\Big]
 \nonumber\\
 &\leq (1-2\beta)^{-1}\sqrt{\frac{2\ln\mu^{-1}}{k}}.
\label{calc_sqrt_k}
\end{align}
Now, assume that $\tp_{m-1}=h$.
Under this, using~\eqref{H_Hoef}
 and~\eqref{calc_sqrt_k}, we obtain
 by conditioning on the value of~$X_m$
\begin{align}
 \IP[\hp_m \in (2\mu,1-2\mu)] &=
  \IE\IP[\hp_m\in (2\mu,1-2\mu)\mid X_m]
  \nonumber\\
  &= \IE\big(\IP[\hp_m\in (2\mu,1-2\mu)\mid X_m]
         \1{e^{-2k(X_m-h)^2}< \mu}
    \nonumber\\      
   &\qquad  + \IP[\hp_m\in (2\mu,1-2\mu)\mid X_m]
         \1{e^{-2k(X_m-h)^2}\geq \mu}\big)
  \nonumber\\         
  &\leq \IE\big(\IP[\hp_m>2\mu\mid X_m]
         \1{e^{-2k(X_m-h)^2}< \mu,h<X_m}
    \nonumber\\      
    &\qquad + \IP[\hp_m<1-2\mu\mid X_m]
         \1{e^{-2k(X_m-h)^2}< \mu,h>X_m}
    \nonumber\\      
   &\qquad  + 
         \1{e^{-2k(X_m-h)^2}\geq \mu}\big)
  \nonumber\\   
  &\leq 2\IP\big[((1-q)n)^{-1}
   {\cal B}((1-q)n,e^{-2k(1-\mu)})<1-2\mu\big]
  \nonumber\\   
   &\qquad    + \IP[e^{-2k(X_m-h)^2}\geq \mu]
     \nonumber\\  
   & \leq \psi_{cau}(n,k),
  \label{long_c1}
\end{align}
recall~\eqref{df_psi_cautious}. This implies
the first comparison in~\eqref{eq_dominate_Psi}.

For a berserk adversary, the calculation is quite
analogous (recall Figure~\ref{f_cautious_berserk}),
so we omit it.
\end{proof}
Next, we need a result that shows that if one 
of the opinions has already reached a supermajority,
then this situation is likely to be preserved.

\begin{figure}
\begin{center}
\includegraphics{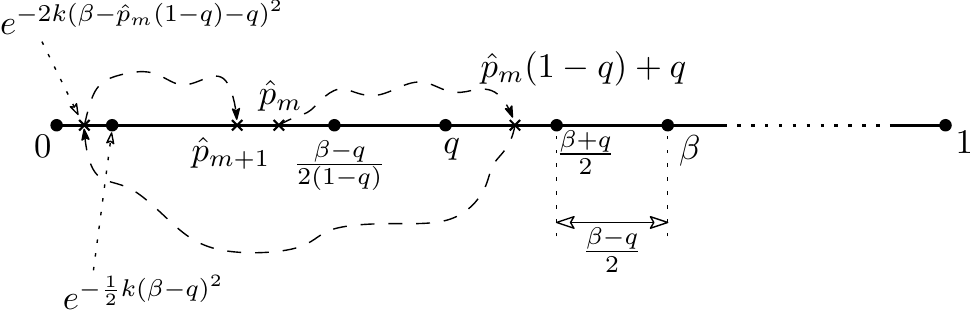}
\caption{Transition from $\hp_m$ to $\hp_{m+1}$: after $m$th
round, being $\hp_m\leq \frac{\beta-q}{2(1-q)}$, 
the adversary may ``grow'' the proportion of $1$s 
to $\hp_m(1-q)+q\leq \frac{\beta+q}{2}$. Then, since the 
difference between that and 
``the least possible threshold''~$\beta$
is at least~$\frac{\beta-q}{2}$, the probability
that an undecided node would have opinion~$1$ in the next round is 
at most~$e^{-\frac{1}{2}k(\beta-q)^2}$. Then,
with overwhelming probability~$\hp_{m+1}$
will be at most~$\frac{\beta-q}{2(1-q)}$,
and so it goes.}
\label{f_beta_q_k}
\end{center}
\end{figure}

\begin{lem}
\label{l_pm_pm+1}
 Let $m\geq 2$; in the following, $A$ will denote a 
subset of~$\{1,\ldots,(1-q)n\}$.
\begin{itemize}
 \item[(i)] Let~$G_0$ be the event 
that $\hp_m\leq \frac{\beta-q}{2(1-q)}$,
$R_{m-1}=A$, and $\hxi_{m-1}(j)=0$ for all~$j\in A$.
Then
\begin{equation}
\label{eq_pm_pm+1_0}
 \IP\Big[\hp_{m+1}\leq \frac{\beta-q}{2(1-q)}\; \big|\; G_0\Big] 
   \geq 1 - e^{-2(1-q)n\varphi_{\beta,q,k}^2}.
\end{equation}
 \item[(ii)] Let~$G_1$ be the event 
that $\hp_m\geq 1-\frac{\beta-q}{2(1-q)}$,
$R_{m-1}=A$, and $\hxi_{m-1}(j)=1$ for all~$j\in A$.
Then
\begin{equation}
\label{eq_pm_pm+1_1}
 \IP\Big[\hp_{m+1}\geq 1-\frac{\beta-q}{2(1-q)}\; \big|\; G_1\Big] 
   \geq 1 - e^{-2(1-q)n\varphi_{\beta,q,k}^2}.
\end{equation}
\end{itemize}
\end{lem}

\begin{proof}
 We prove only part~(i), the proof of the other part is
completely analogous. Now, look at Figure~\ref{f_beta_q_k}:
essentially, this is a direct 
consequence of Lemma~\ref{l_LD_fixed} with 
$\theta=e^{-2k(\frac{\beta-q}{2})^2}
=e^{-\frac{1}{2}k(\beta-q)^2}$ and $v=\varphi_{\beta,q,k}$.
Observe also that, if some honest nodes already decided 
on~$0$ definitely, it holds that 
$(1-q)n\hp_{m+1}$ is stochastically dominated
by $\mathcal{B}\big((1-q)n,e^{-\frac{1}{2}k(\beta-q)^2}\big)$.
\end{proof}

Now, we are able to conclude the proof of Theorem~\ref{t_safety_liveness}.
Let us introduce the random variable
\begin{equation}
\label{def_Z}
 Z = 
  \begin{cases}
   \min\Big\{m > \Psi:  \hp_m>\frac{\beta-q}{2(1-q)}\Big\} & \text{on }
                    \hp_\Psi \leq \frac{\beta-q}{2(1-q)},\\
   \min\Big\{m > \Psi: \hp_m<1-\frac{\beta-q}{2(1-q)}\Big\} 
   \vphantom{\int\limits^{A^A}} & \text{on }
                    \hp_\Psi \geq 1-\frac{\beta-q}{2(1-q)}              
  \end{cases}
\end{equation}
to be the first moment after~$\Psi$ when the honest nodes'
opinion has drifted away from supermajority.
Denote also
\[
 \htau^{(1)}_j = \min\big\{m\geq m_0+\ell : 
  \hxi_m^{(j)}=\hxi_{m-1}^{(j)}=\ldots 
  = \hxi_{m-\ell+1}^{(j)}=1\big\}
\]
and
\[
 \htau^{(0)}_j = \min\big\{m\geq m_0+\ell : 
  \hxi_m^{(j)}=\hxi_{m-1}^{(j)}=\ldots 
  = \hxi_{m-\ell+1}^{(j)}=0\big\}.
\]

Next, observe that
\[
  (H_0\cup H_1) \cap \{\mathcal{N}\leq m_0+\ell u\}
 \subset D_1 \cap D_2 \cap D_3, 
\]
where
\begin{align*}
 D_1 &= \{\Psi\leq m_0\},\\
 D_2 &= \{Z\geq m_0+\ell u\},\\
 D_3 &= \big\{ \text{there is }
    i\in\{0,1\} 
 \text{ such that }
   \htau^{(i)}_j \leq m_0+\ell u,
   \htau^{(i)}_j<\htau^{(1-i)}_j\\
  & \qquad \qquad \qquad\qquad \qquad\qquad\qquad\qquad
  \text{ for all }
    j=1,\ldots,(1-q)n\big\} .
\end{align*}
To obtain the estimates~\eqref{eq_safety_liveness_cautious} 
and~\eqref{eq_safety_liveness_berserk},
it is enough to note that the lower bounds
on, respectively, $\IP[D_1]$, $\IP[D_2]$, and $\IP[D_3]$,
follow from, respectively, Lemma~\ref{l_dominate_Psi},
Lemma~\ref{l_pm_pm+1}, and Lemma~\ref{l_elem_calc}
(and also the union bound).
\end{proof}

\begin{proof}[Proof of Theorem~\ref{t_initial}.]
We prove only the part~(i); the proof of the other
part is completely analogous.
In fact, to obtain the proof it is enough to observe that,
if $a-\hp_0(1-q)-q>0$ and 
$e^{-2k(a-\hp_0 (1-q) - q)^2} \leq \frac{\beta-q}{4(1-q)} $, 
then, by~\eqref{H_Hoef}, 
with probability at least $1- \exp\big(-\tfrac{1}{8}n
     \tfrac{(\beta-q)^2}{4(1-q)}\big)$ it happens that
    $\hp_1\leq \frac{\beta-q}{2(1-q)}$ 
(so, in particular, $\Psi=1$); next, the same argument
as in the proof of Theorem~\ref{t_safety_liveness}
does the work.
\end{proof}

\section{Further generalizations}
\label{s_gener}
In this section we argue that our protocol is \emph{robust},
that is, it is possible to adapt it in such a way that it is
able to work well in more ``practical''
situations. Specifically, observe that 
nodes may not always respond queries,
and the adversarial nodes sometimes may do so \emph{deliberately}.
The protocol described in Section~\ref{s_descr}
is not designed to handle this, so it needs to be amended.
There are at least two natural ways to deal 
with this situation:
\begin{itemize}
    \item[(i)] let each node to take the decision 
    based on the responses that it effectively received 
    (i.e., instead of $k^{-1}\eta_m(j)$ use 
    $\eta_m(j)/ \zeta_m(j) $, where $\zeta_m(j)$
    is the number of responses that the $j$th node 
    received in the $m$th round);
    \item[(ii)] each node queries more than~$k$ nodes,
    say, $2k$ or more; since whp the number of responses 
    received will be at least~$k$ (for
    definiteness, let us assume that the probability
    that a query is left unresponded is less than $\frac{1}{2}$),
    the node then keeps exactly~$k$ responses and discards the rest;
\end{itemize}
and it is of course also possible to combine them.
  The practical difference between these two options is 
 probably not so big at least in the case of reasonably large values of~$k$
(because then the proportions of $1$s in the responses should be roughly the same in most cases due to the Law of Large Numbers);
 for the sake of formulating the results 
 in a more clean way, let us assume
 that a node simply issues queries sequentially
 until getting exactly~$k$ responses.
    

Now, we define the notion of a \emph{semi-cautious}
adversary: every node it controls will not give 
\emph{contradicting} responses (i.e., $0$ to one node
and~$1$ to another node in the same round)
but can sometimes remain silent; since it does not make
sense for a node to remain silent altogether in a given round
(that would just reduce the fraction of the adversarial
nodes in the network), there are two possible adversarial
node behaviours:
\begin{itemize}
    \item a node answers~``$0$'' to some queries 
    and does not answer other queries;
    \item a node answers~``$1$'' to some queries 
    and does not answer other queries.
\end{itemize}
Here is the result we have for a semi-cautious adversary:
\begin{theo}
\label{t_safety_liveness_semi}
 If the adversary is semi-cautious,
assume that $\frac{1}{2-q}-\beta+\sqrt{2k^{-1}
\ln\tfrac{4(1-q)}{\beta-q}}<1-2\beta$.
Then, for any adversarial strategy, we have 
\begin{align}
\IP\big[(H_0\cup H_1)\cap\{\mathcal{N}\leq m_0+\ell u\}\big]
 &  \geq 1 - \big(\psi_{semi}(n,k)\big)^{m_0}
 -W(n,k,m_0,\ell,u),
\label{eq_safety_liveness_semi}  
\end{align}
where
\begin{equation}
\label{df_vr_theta_nk}
\psi_{semi}(n,k)
= 2\exp\big(-\tfrac{1}{8}n\tfrac{(\beta-q)^2}{4(1-q)}\big) + 
\frac{\frac{1}{2-q}-\beta+\sqrt{2k^{-1}
\ln\tfrac{4(1-q)}{\beta-q}}}{1-2\beta}.
\end{equation}
\end{theo}

In this situation, the fact corresponding 
to Corollary~\ref{c_1/2_1/3} will be
the following (in particular, note the new 
``security threshold'' $\phi^{-2}\in(\frac{1}{3},\frac{1}{2})$
that we obtain here):
\begin{cor}
\label{c_Golden}
 For a semi-cautious adversary, we need that 
$\frac{1}{2-q}-\beta<1-2\beta$, or, equivalently, 
$q<2-\frac{1}{1-\beta}$
(it is only in this case that we will be able
to find large enough~$k$ such that the hypothesis
of Theorem~\ref{t_safety_liveness_semi}
is satisfied). Since we also still need 
 that $q<\beta$, solving
$\beta = 2-\frac{1}{1-\beta}$, 
we obtain that~$q$ must be less 
than~$\frac{3-\sqrt{5}}{2}=\frac{1}{1+\phi}=\phi^{-2}\approx0.38$,
where
 $\phi=\frac{1+\sqrt{5}}{2}$ is the Golden Ratio.
Then, as before, it is straightforward to 
show that, for a semi-cautious adversary, 
for any~$q<\frac{1}{1+\phi}$ and all large enough~$n$ 
 we are able to adjust the parameters $k,\beta,m_0,\ell$
 in such a way
that the protocol works whp (in particular, 
a $\beta$-value sufficiently close to~$\frac{1}{1+\phi}$ would work).
\end{cor}

\begin{proof}[Proof of Theorem~\ref{t_safety_liveness_semi}] 
As observed before,
an ``always-silent'' strategy is not interesting
for an adversarial node, since this will, in practice,
only reduce
their quantity. Now, assume that, for some $\gamma\in [0,1]$,
\begin{itemize}
    \item $\gamma q n$ adversarial nodes reply ``$0$''
    or remain silent;
    \item $(1-\gamma) q n$ adversarial nodes reply ``$1$''
    or remain silent.
\end{itemize}
Then, if the adversary wants to 
decrease a honest node's confidence in the $1$-opinion,
those nodes who may answer ``$1$'' will remain silent,
and so
with probability $\frac{1-q}{1-q+\gamma q}$
the response will be obtained from a honest node, while
with probability $\frac{\gamma q}{1-q+\gamma q}$
the response will be obtained from an adversarial node.
This gives 
\[
 \hp \frac{1-q}{1-q+\gamma q}
  = \frac{\hp(1-q)}{1-(1-\gamma)q}
\]
as the ``lower limit'' for the (expected) proportion
of~$1$s in the queries.
Analogously, if the adversary wants to 
increase an honest node's confidence in the $1$-opinion
those nodes who may answer ``$0$'' will remain silent, and so
with probability $\frac{1-q}{1-q+(1-\gamma) q}$
the response will be obtained from a honest node, while
with probability $\frac{(1-\gamma) q}{1-q+(1-\gamma) q}$
the response will be obtained from an adversarial node.
This gives 
\[
 \hp \frac{1-q}{1-q+(1-\gamma) q} 
   + \frac{(1-\gamma)q}{1-q+(1-\gamma) q}
  = \frac{\hp(1-q)+(1-\gamma)q}{1-\gamma q}
\]
as the corresponding ``upper limit''.
So, analogously to Figure~\ref{f_cautious_berserk},
the semi-cautious adversary can achieve
the ``crosses'' to be distributed on the interval
\begin{equation}
\label{df_IIgamma}   
\II_\gamma := \Big[\frac{\hp(1-q)}{1-(1-\gamma)q},
 \frac{\hp(1-q)+(1-\gamma)q}{1-\gamma q}\Big]
\end{equation}
in any way. 
Now, it is elementary to see that both endpoints
of the above interval decrease when~$\gamma$ increases;
if we want to make it symmetric (around~$\frac{1}{2}$),
we need to solve 
\[
 \frac{\hp(1-q)}{1-(1-\gamma)q}
 = 1 -  \frac{\hp(1-q)+(1-\gamma)q}{1-\gamma q},
\]
or, equivalently
\[
  \frac{\hp}{1-(1-\gamma)q}
 =  \frac{1-\hp}{1-\gamma q}
\]
for~$\gamma$. This gives the solution
$\gamma^* = q^{-1}(2\hp-1)+(1-\hp)$.
After substituting~$\gamma^*$ to~\eqref{df_IIgamma}, 
the symmetrized interval becomes
\[
\II_{\gamma^*}=\Big[\frac{1-q}{2-q},
 \frac{1}{2-q}\Big]
\]
(somewhat unexpectedly, because it doesn't depend on~$\hp$
anymore).
\begin{figure}
\begin{center}
\includegraphics{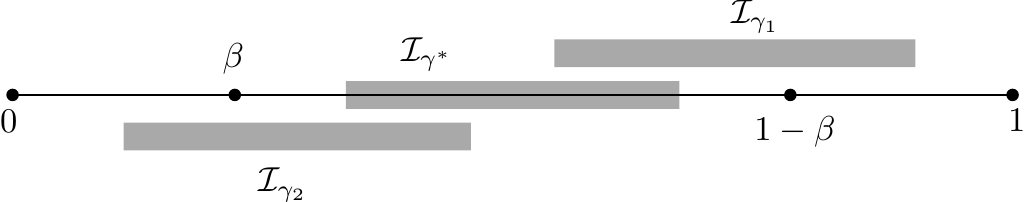}
\caption{The ``intervals of control'' of a semi-cautious
adversary, for $\gamma_1<\gamma^*<\gamma_2$}
\label{f_moving_int}
\end{center}
\end{figure}
It is actually worth noting that~$\gamma^*$ does not necessarily
belong to~$[0,1]$
(so it is not always possible to make this interval symmetric), 
but it does not pose a problem due to 
the following. Look at Figure~\ref{f_moving_int}:
due to the monotonicity, 
\begin{equation}
\label{max_intersection}
 \max_{\gamma\in[0,1]}
  \big|[\beta,1-\beta]\cap\II_\gamma\big| \leq 
   \frac{1}{2-q}-\beta ;
\end{equation}
indeed, for all~$\gamma$ we see that 
either the interval~$[\beta,\frac{1-q}{2-q})$
or the interval~$(\frac{1}{2-q}, 1-\beta]$
is a subset of $[\beta,1-\beta]\setminus \II_\gamma$.

This essentially takes care of the argument in the proof
of Lemma~\ref{l_dominate_Psi} (since we now understand 
what is the minimal length of the interval that the adversary
cannot control),
and the rest of the proof
is completely analogous to that of Theorem~\ref{t_safety_liveness}:
indeed, as observed before,
the adversary loses control after the random time~$\Psi$.
\end{proof}
We also observe that Theorem~\ref{t_initial} remains
valid also for a semi-cautious adversary.

Next, let us discuss what do we really need from the 
(decentralized) random number generator. 
In fact, it is not so much:
we need that, regardless of the past, 
with probability at least~$\theta$ (where $\theta>0$ is a fixed parameter)
the next outcome is a uniform random variable
which is ``unpredictable'' for the adversary; 
this random number is seen by at least 
(1-$\delta$) proportion of honest nodes, 
where~$\delta$ is reasonably small. 
What we can prove in such a situation depends on what
the remaining $\delta(1-q)n$ honest nodes use as their
decision thresholds: they can use some ``second candidate''
(in case there is an alternative 
source of common randomness), or they can choose their 
thresholds independently and randomly, etc.
Each of such situations would need to be treated separately,
which is certainly doable, but left out of this paper.
Let us note, though, that
the ``worst-case'' assumption 
is that the adversary can ``feed'' the (fake) decision
thresholds to those $\delta(1-q)n$ honest
nodes. This 
would effectively mean that these 
nodes would behave as cautious adversaries in the next round
(which matters if the random time~$\Psi$ did not yet occur).
Therefore, for the sake of obtaining 
bounds like~\eqref{eq_safety_liveness_cautious}--\eqref{eq_safety_liveness_berserk} 
and~\eqref{eq_safety_liveness_semi} we 
can simply pretend that the value of~$q$ is increased by~$\delta$.

Now, assuming that $\delta=0$, it is easy to
figure out how this will affect our results: indeed,
in our proofs, all random thresholds matter only until~$\Psi$.
It is then straightforward to obtain the following fact:
\begin{prop}
\label{p_weaker_RNG}
Assume the above on the random number generation
(with $\theta\in (0,1)$ and $\delta=0$). Then, 
the estimates 
\eqref{eq_safety_liveness_cautious}--\eqref{eq_safety_liveness_berserk} 
and~\eqref{eq_safety_liveness_semi}
remain valid with $1-\theta+\theta\psi_{*}(n,k)$
on the place of $\psi_{*}(n,k)$
(with $* \in \{\text{cau}, \text{ber}, \text{semi}\}$).
\end{prop}
In view of the above result,
let us stress that
one of the main ideas of this paper is: 
we use a ``rather weak'' consensus (on the random numbers,
as above) to obtain a ``strong'' consensus on the value of a bit 
(i.e., validity of a transaction). 
Also, let us observe that a
partial control of the random numbers does not give 
access to a lot of power (in the worst case the adversary 
would delay the consensus a bit, but that is all), 
so there is not much need to be restrictive on the degree 
of decentralization for that part\footnote{in
other words, it may make sense that different parts 
of the system are decentralized to a different degree, there is nothing a priory wrong with it}: 
a smaller subcommittee can take care of the random numbers'
generation, and some 
VDF-based random number generation scheme 
(such as~\cite{Lenstra_Wes17})
may be used
to further prevent this subcommittee from leaking the numbers 
before the due time).

\section{Notations used in this paper}
\label{s_notations}
Here, for reader's convenience, we provide a list of most important notations used 
in the paper.
\begin{itemize}
 \item $n\in\N$ is the total number of nodes in the system, $k\in\N$ is the number 
of queries each node makes in one round, $q\in [0,1)$ is the proportion of adversarial 
nodes;
 \item $1/2 < a \leq b <1$ are the threshold limits in the first round,
 $\beta\in (0,1/2)$ is the threshold limit parameter in the subsequent rounds;
 \item in the same round, 
 a \emph{cautious} adversarial node always responds the queries in the same way,
 a \emph{berserk} one can give any responses (including not responding at all), 
 and a \emph{semi-cautious} one may ignore some of the queries but has to respond
 the others in the same way;
 \item for our protocol to work, for all cases we need to assume that $q<\beta$;
 additionally, we also need to assume that $q<1-2\beta$ for berserk adversary 
 and $q<2-\frac{1}{1-\beta}$ for semi-cautious adversary;
 \item $m_0\in\N$, the cooling-off period, $\ell\in\N$ is the number of consecutive rounds
 with the same opinion after which it becomes final;
 \item $\eta_m(j)\in\{1,\ldots,k\}$ is the number of $1$-opinions that $j$th node
 receives in $m$th round;
 \item $U[s,r]$ is the Uniform probability distribution of the interval $[s,r]$ 
 and $\mathcal{B}(m,p)$ is the Binomial distribution 
with parameters $m\geq 1$ and $p\in [0,1]$;
 \item $W(\cdot,\cdot,\cdot,\cdot,\cdot)$ is defined in~\eqref{df_bigW},
 $\psi_{cau}(\cdot,\cdot)$ and $\psi_{ber}(\cdot,\cdot)$ are defined 
 in~\eqref{df_psi_cautious}--\eqref{df_psi_berserk}, and 
 $\psi_{semi}(\cdot,\cdot)$ is defined in~\eqref{df_vr_theta_nk};
 \item $H_0\cap H_1$ is the ``consensus event'' (i.e., either all honest nodes 
 finalize on~$0$ or all honest nodes finalize on~$1$);
 \item $\mathcal{N}$ is the number of rounds until 
all honest nodes finalize their opinions;
 \item $\hp_0$ is the initial proportion of $1$-opinions
among the honest nodes;
 \item $\hp_m$ is the proportion of $1$-opinions
among the honest nodes after the $m$th round;
 \item $R_m\subset \{1,\ldots,n\}$ is the set of honest nodes that finalized
their opinions by round~$m$;
 \item $\Psi\in\N$ is the (random) moment when the adversary ``looses control''  
 (formally defined in~\eqref{def_Psi}).
\end{itemize}

\section{Conclusions and Future Work}
\label{s_final}
\noindent In this paper we described a consensus protocol
which is able to withstand a substantial proportion
of Byzantine nodes, and obtained some explicit
estimates on its safety and liveness.
A special feature of our protocol is that
it uses a sequence of random numbers (produced
by some external source or by the nodes themselves)
in order to have a ``randomly moving decision
threshold'' which quickly defeats the adversary's
attempts to mess with the consensus.
It is also worth noting that
 the ``quality'' of those random numbers is 
not critically important --
only the estimates on~$\Psi$ (Lemma~\ref{l_dominate_Psi})
will be affected in a non-drastic way.
In particular, one can permit that the random numbers
might be biased, or even that the adversary 
might get control of these numbers from time to time.
Also, it is clear from the proofs that there is no need
for the honest nodes to achieve consensus on the actual
values of these random numbers: if some (not very large) proportion
of honest nodes does not see the same number as the others,
this will not cause problems.
All this is due to the fact that, when the proportion
of $1$-opinions among the honest nodes becomes
``too small'' or ``too large'' (i.e., 
less than $\frac{\beta-q}{2(1-q)}$ or greater
than $1-\frac{\beta-q}{2(1-q)}$ in our proofs), the adversary
does not have any control anymore.


As mentioned above, this paper primarily contains a rigorous analysis of the simplest version of the protocol. It is of course tempting to 
consider other versions, in particular,
where the neighborhood relation between nodes is not that of a complete graph. 
However, as is frequently the case, introducing a nontrivial graph structure 
makes the problem hardly tractable in a purely analytic way, which means
that it has to be approached via simulations. For that approach,
we refer to the paper~\cite{capossele2019}, which can be considered as a 
complementary work to the present one (note that simulations also permit 
one to obtain better estimates the failure probabilities since, as noted above,
the estimates~\eqref{eq_safety_liveness_cautious} 
and~\eqref{eq_safety_liveness_berserk} of the present paper 
are usually not quite sharp).

Also, when considering the behavior of the adversary nodes, a number
of ``worst-case'' assumptions were made (notably, that all the adversarial nodes
are controlled by a single entity and that the adversary is omniscient).
Again, analysing a system with independent adversaries with possibly 
incomplete information is an interesting but analytically difficult program;
in this case the simulations approach seems to be more adequate as well.

We need to comment on anti-Sybil measure in practical 
implementations: indeed, it would be quite unfortunate
if the adversary is able to deploy an excessively 
large number of nodes, thus inflating the value of~$q$.
One of the possible approaches is using a 
variant of Proof-of-Stake;
with it, when querying, one needs to choose the node proportionally
to its weight (stake). This is partly the subject of an 
ongoing research effort~\cite{muller2020new}.


\bibliographystyle{IEEEtran}
\bibliography{main}

\begin{thebibliography}{10}
\providecommand{\url}[1]{#1}
\csname url@samestyle\endcsname
\providecommand{\newblock}{\relax}
\providecommand{\bibinfo}[2]{#2}
\providecommand{\BIBentrySTDinterwordspacing}{\spaceskip=0pt\relax}
\providecommand{\BIBentryALTinterwordstretchfactor}{4}
\providecommand{\BIBentryALTinterwordspacing}{\spaceskip=\fontdimen2\font plus
\BIBentryALTinterwordstretchfactor\fontdimen3\font minus
  \fontdimen4\font\relax}
\providecommand{\BIBforeignlanguage}[2]{{%
\expandafter\ifx\csname l@#1\endcsname\relax
\typeout{** WARNING: IEEEtran.bst: No hyphenation pattern has been}%
\typeout{** loaded for the language `#1'. Using the pattern for}%
\typeout{** the default language instead.}%
\else
\language=\csname l@#1\endcsname
\fi
#2}}
\providecommand{\BIBdecl}{\relax}
\BIBdecl

\bibitem{zheng2017overview}
Z.~Zheng, S.~Xie, H.~Dai, X.~Chen, and H.~Wang, ``An overview of blockchain
  technology: Architecture, consensus, and future trends,'' in \emph{2017 IEEE
  International Congress on Big Data (BigData Congress)}.\hskip 1em plus 0.5em
  minus 0.4em\relax IEEE, 2017, pp. 557--564.

\bibitem{aguilera2012correctness}
M.~K. Aguilera and S.~Toueg, ``The correctness proof of {B}en-{O}r’s
  randomized consensus algorithm,'' \emph{Distributed Computing}, vol.~25,
  no.~5, pp. 371--381, 2012.

\bibitem{ben1983another}
M.~Ben-Or, ``Another advantage of free choice: Completely asynchronous
  agreement protocols (extended abstract),'' in \emph{Proceedings of the 2nd
  ACM Annual Symposium on Principles of Distributed Computing, Montreal,
  Quebec}, 1983, pp. 27--30.

\bibitem{bracha1987asynchronous}
G.~Bracha, ``Asynchronous {B}yzantine agreement protocols,'' \emph{Information
  and Computation}, vol.~75, no.~2, pp. 130--143, 1987.

\bibitem{feldman1989optimal}
P.~Feldman and S.~Micali, ``An optimal probabilistic algorithm for synchronous
  {B}yzantine agreement,'' in \emph{International Colloquium on Automata,
  Languages, and Programming}.\hskip 1em plus 0.5em minus 0.4em\relax Springer,
  1989, pp. 341--378.

\bibitem{friedman2005simple}
R.~Friedman, A.~Mostefaoui, and M.~Raynal, ``Simple and efficient oracle-based
  consensus protocols for asynchronous {B}yzantine systems,'' \emph{IEEE
  Transactions on Dependable and Secure Computing}, vol.~2, no.~1, pp. 46--56,
  2005.

\bibitem{rabin1983randomized}
M.~O. Rabin, ``Randomized {B}yzantine generals,'' in \emph{24th Annual
  Symposium on Foundations of Computer Science (sfcs 1983)}.\hskip 1em plus
  0.5em minus 0.4em\relax IEEE, 1983, pp. 403--409.

\bibitem{liu2018scalable}
J.~Liu, W.~Li, G.~O. Karame, and N.~Asokan, ``Scalable {B}yzantine consensus
  via hardware-assisted secret sharing,'' \emph{IEEE Transactions on
  Computers}, vol.~68, no.~1, pp. 139--151, 2018.

\bibitem{crain2018dbft}
T.~Crain, V.~Gramoli, M.~Larrea, and M.~Raynal, ``Dbft: Efficient leaderless
  {B}yzantine consensus and its application to blockchains,'' in \emph{2018
  IEEE 17th International Symposium on Network Computing and Applications
  (NCA)}.\hskip 1em plus 0.5em minus 0.4em\relax IEEE, 2018, pp. 1--8.

\bibitem{miller2016honey}
A.~Miller, Y.~Xia, K.~Croman, E.~Shi, and D.~Song, ``The honey badger of bft
  protocols,'' in \emph{Proceedings of the 2016 ACM SIGSAC Conference on
  Computer and Communications Security}.\hskip 1em plus 0.5em minus 0.4em\relax
  ACM, 2016, pp. 31--42.

\bibitem{holley1975ergodic}
R.~A. Holley, T.~M. Liggett \emph{et~al.}, ``Ergodic theorems for weakly
  interacting infinite systems and the voter model,'' \emph{The annals of
  probability}, vol.~3, no.~4, pp. 643--663, 1975.

\bibitem{CliffSud}
\BIBentryALTinterwordspacing
P.~Clifford and A.~Sudbury, ``A model for spatial conflict,''
  \emph{Biometrika}, vol.~60, no.~3, pp. 581--588, 1973. [Online]. Available:
  \url{http://www.jstor.org/stable/2335008}
\BIBentrySTDinterwordspacing

\bibitem{becchetti2016stabilizing}
L.~Becchetti, A.~Clementi, E.~Natale, F.~Pasquale, and L.~Trevisan,
  ``Stabilizing consensus with many opinions,'' in \emph{Proceedings of the
  twenty-seventh annual ACM-SIAM symposium on Discrete algorithms}.\hskip 1em
  plus 0.5em minus 0.4em\relax SIAM, 2016, pp. 620--635.

\bibitem{cooper2014power}
C.~Cooper, R.~Els{\"a}sser, and T.~Radzik, ``The power of two choices in
  distributed voting,'' in \emph{International Colloquium on Automata,
  Languages, and Programming}.\hskip 1em plus 0.5em minus 0.4em\relax Springer,
  2014, pp. 435--446.

\bibitem{cooper2015fast}
C.~Cooper, R.~Els{\"a}sser, T.~Radzik, N.~Rivera, and T.~Shiraga, ``Fast
  consensus for voting on general expander graphs,'' in \emph{International
  Symposium on Distributed Computing}.\hskip 1em plus 0.5em minus 0.4em\relax
  Springer, 2015, pp. 248--262.

\bibitem{elsasser2016rapid}
R.~Els{\"a}sser, T.~Friedetzky, D.~Kaaser, F.~Mallmann-Trenn, and H.~Trinker,
  ``Rapid asynchronous plurality consensus,'' \emph{arXiv preprint
  arXiv:1602.04667}, 2016.

\bibitem{fanti2019communication}
G.~Fanti, N.~Holden, Y.~Peres, and G.~Ranade, ``Communication cost of consensus
  for nodes with limited memory,'' \emph{arXiv preprint arXiv:1901.01665},
  2019.

\bibitem{CruiseGanesh14}
\BIBentryALTinterwordspacing
J.~Cruise and A.~Ganesh, ``Probabilistic consensus via polling and majority
  rules,'' \emph{Queueing Syst. Theory Appl.}, vol.~78, no.~2, pp. 99--120,
  Oct. 2014. [Online]. Available:
  \url{http://dx.doi.org/10.1007/s11134-014-9397-7}
\BIBentrySTDinterwordspacing

\bibitem{aumann2007security}
Y.~Aumann and Y.~Lindell, ``Security against covert adversaries: Efficient
  protocols for realistic adversaries,'' in \emph{Theory of Cryptography
  Conference}.\hskip 1em plus 0.5em minus 0.4em\relax Springer, 2007, pp.
  137--156.

\bibitem{cascudo2017scrape}
I.~Cascudo and B.~David, ``Scrape: Scalable randomness attested by public
  entities,'' in \emph{International Conference on Applied Cryptography and
  Network Security}.\hskip 1em plus 0.5em minus 0.4em\relax Springer, 2017, pp.
  537--556.

\bibitem{Lenstra_Wes17}
\BIBentryALTinterwordspacing
A.~K. Lenstra and B.~Wesolowski, ``Trustworthy public randomness with sloth,
  unicorn, and trx,'' \emph{International Journal of Applied Cryptography},
  vol.~3, no.~4, pp. 330--343, 2017. [Online]. Available:
  \url{https://www.inderscienceonline.com/doi/abs/10.1504/IJACT.2017.089354}
\BIBentrySTDinterwordspacing

\bibitem{popov2017decentralized}
S.~Popov, ``On a decentralized trustless pseudo-random number generation
  algorithm,'' \emph{Journal of Mathematical Cryptology}, vol.~11, no.~1, pp.
  37--43, 2017.

\bibitem{schindlerhydrand}
P.~Schindler, A.~Judmayer, N.~Stifter, and E.~Weippl, ``Hydrand: Practical
  continuous distributed randomness,'' Cryptology ePrint Archive, Report
  2018/319, 2018, \url{https://eprint.iacr.org/2018/319}.

\bibitem{syta2017scalable}
E.~Syta, P.~Jovanovic, E.~K. Kogias, N.~Gailly, L.~Gasser, I.~Khoffi, M.~J.
  Fischer, and B.~Ford, ``Scalable bias-resistant distributed randomness,'' in
  \emph{2017 IEEE Symposium on Security and Privacy (SP)}.\hskip 1em plus 0.5em
  minus 0.4em\relax IEEE, 2017, pp. 444--460.

\bibitem{rocket2019ava}
T.~Rocket, M.~Yin, K.~Sekniqi, R.~van Renesse, and E.~G. Sirer, ``Scalable and
  probabilistic leaderless bft consensus through metastability,'' 2019.

\bibitem{hoeffding1994probability}
W.~Hoeffding, ``Probability inequalities for sums of bounded random
  variables,'' in \emph{The Collected Works of Wassily Hoeffding}.\hskip 1em
  plus 0.5em minus 0.4em\relax Springer, 1994, pp. 409--426.

\bibitem{capossele2019}
A.~Capossele, S.~M\"uller, and A.~Penzkofer, ``Robustness and efficiency of
  leaderless probabilistic consensus protocols within byzantine
  infrastructures,'' 2019.

\bibitem{muller2020new}
S.~M\"uller, A.~Penzkofer, B.~Ku\'{s}mierz, D.~Camargo, and W.~J. Buchanan,
  ``Fast {P}robabilistic {C}onsensus with {W}eighted {V}otes (work in
  progress),'' 2020.

\end{thebibliography}

\section*{Acknowledgments}
\noindent The author thanks Hans Moog, Sebastian M\"uller,
Luigi Vigneri, and Wolfgang Welz for valuable comments and suggestions and also for providing some simulations of the model.


\appendix
\section{Runs of zeros and ones in Bernoulli trials}
Here we prove a simple fact about runs of zeros and ones 
in a sequence of Bernoulli trials, which will imply~\eqref{runs}.
Namely, let $\xi_1, \xi_2, \xi_3,\ldots$ be i.i.d.\
random variables with $\IP[\xi_i=1]=1-\IP[\xi_i=0]=h$,
and let, for $r,s\in \N$
\[
 \tau = \min\{m\geq r : 
  \xi_m=\xi_{m-1}=\ldots 
  = \xi_{m-r+1}=1\}
\]
and
\[
 \sigma = \min\{m\geq s : 
  \xi_m=\xi_{m-1}=\ldots 
  = \xi_{m-s+1}=0\}
\]
to be the first moments when we see runs of~$r$
ones (respectively, $s$ zeros). 
\begin{prop}
\label{p_elemfact}
It holds that
\begin{equation}
\label{runs_general}
 \IP[\tau<\sigma] = \frac{h^{r-1}(1-(1-h)^s)}
 {h^{r-1}+(1-h)^{s-1}-h^{r-1}(1-h)^{s-1}}.
\end{equation}
\end{prop}

\begin{proof}
To prove~\eqref{runs_general}, we use conditioning.
Abbreviate 
$p_0 = \IP[\tau<\sigma\mid \xi_1=0]$ and
$p_1 = \IP[\tau<\sigma\mid \xi_1=1]$.
Then, conditioning on the number of consecutive zeros
in the beginning, we write
\begin{equation}
\label{p0=Cp1}
 p_0 = \sum_{j=1}^{s-1}(1-h)^{j-1}h p_1 = (1-(1-h)^{s-1})p_1,
\end{equation}
and, conditioning on the number of consecutive ones
in the beginning, we obtain that
\begin{equation}
\label{p1=Cp0}
 p_1 = \sum_{j=1}^{r-1}h^{j-1}(1-h) p_0 
   + \sum_{j=r}^\infty h^{j-1}(1-h)\times 1
 = (1-h^{r-1})p_0 + h^{r-1}.
\end{equation} 
Solving  \eqref{p0=Cp1}--\eqref{p1=Cp0} for $p_{0,1}$
yields
\begin{align*}
p_0 &= \frac{h^{r-1}(1-(1-h)^{s-1})}
 {h^{r-1}+(1-h)^{s-1}-h^{r-1}(1-h)^{s-1}},\\
p_1 &= \frac{h^{r-1}}
 {h^{r-1}+(1-h)^{s-1}-h^{r-1}(1-h)^{s-1}},
\end{align*}
and we then obtain~\eqref{runs_general} by
using the obvious relation $\IP[\tau<\sigma]=(1-h)p_0+hp_1$.
\end{proof}

\end{document}